\documentclass[copyright]{eptcs}
\providecommand{\event}{DCFS 2010} 
\usepackage{breakurl}             
\usepackage{amssymb}
\usepackage[T1]{fontenc}

\def \pref#1{{\mathbf{pref}(#1)}}
\def \strr#1{\sqrt[\,{\textstyle *}\,]{#1}}
\def \bbbn{\mathrm{I\!N}}
\def \infix#1{{\mathbf{infix}(#1)}}
\newcommand*{\qed}{\raisebox{0.5ex}[0ex][0ex]{\framebox[1ex][l]{}}}
\newtheorem{theorem}{Theorem}
\newtheorem{lemma}[theorem]{Lemma}
\newtheorem{proposition}[theorem]{Proposition}
\newtheorem{corollary}[theorem]{Corollary}
\newtheorem{definition}[theorem]{Definition}
\newtheorem{example}[theorem]{Example}
\newenvironment{proof}{\textsl{Proof: }}{\hfill\qed\par}

\title{The Maximal Subword Complexity of Quasiperiodic Infinite Words
\author{Ronny Polley
\institute{Martin-Luther-Universit\"at
  Halle-Wittenberg}
\institute{Institut f\"ur Informatik\\
D-06099 Halle (Saale), Germany}
\email{ronny.polley@informatik.uni-halle.de}
\and
Ludwig Staiger
\institute{Martin-Luther-Universit\"at
  Halle-Wittenberg}
\institute{Institut f\"ur Informatik\\
D-06099 Halle (Saale), Germany}
\email{ludwig.staiger@informatik.uni-halle.de}
}
}
\def\titlerunning{Subword Complexity of Quasiperiodic Words}
\def\authorrunning{R. Polley, L. Staiger}

\begin{document}
\maketitle

\begin{abstract}
  We provide an exact estimate on the maximal subword complexity for
  quasiperiodic infinite words. To this end we give a representation of the
  set of finite and of infinite words having a certain quasiperiod $q$ via a
  finite language derived from $q$. It is shown that this language is a suffix
  code having a bounded delay of decipherability. 

  Our estimate of the subword complexity now follows from this result,
  previously known results on the subword complexity and elementary results on
  formal power series.

\end{abstract}
In his tutorial \cite{DBLP:journals/eatcs/Marcus04} Solomon Marcus provided
some initial facts on quasiperiodic infinite words. Here he posed several
questions on the complexity of quasiperiodic infinite words. Some answers
mainly for questions concerning quasiperiodic infinite words of low complexity
were given in \cite{DBLP:journals/eatcs/LeveR04,DBLP:journals/tcs/LeveR07}.

The investigations of the present paper turn to the question which are the
maximally possible complexity functions for those words. As complexity we
follow Marcus'~\cite{DBLP:journals/eatcs/Marcus04} Question~2 to consider the
(subword) complexity function $f(\xi,n)$ of an infinite word $\xi$; $f(\xi,n)$
being its number of subwords of length $n$. This subword complexity of
infinite words ($\omega$-words) was mainly investigated for those words of low
(polynomial) complexity (see the tutorial 
\cite{DBLP:journals/eatcs/BerstelK03} or the book \cite{book:AlloucheSha}). In
\cite{DBLP:journals/iandc/Staiger93,DBLP:conf/csl/Staiger97} some
results on exponential subword complexity helpful for the present
considerations are derived.

As a final result we obtain that the maximally possible complexity functions
for quasiperiodic infinite words $\xi$ are bounded from above by a function of
the form $f(\xi,n)\leq c_{\xi}\cdot t_{P}^{n}$ where $t_{P}$ is the smallest
Pisot-Vijayaraghavan number, that is, the unique real root $t_P$ of the cubic
polynomial $x^{3}-x-1$, which is approximately equal to $t_P\approx
1.324718$. We show also that this bound is tight, that is, there are
$\omega$-words $\xi$ having $f(\xi,n)\approx c\cdot t_{P}^{n}$.

The paper is organised as follows. After introducing some notation we derive
in Section~\ref{sec.quas} a characterisation of quasiperiodic words and
$\omega$-words having a certain quasiperiod $q$. Moreover, we introduce a
finite basis set $P_{q}$ from which the sets of quasiperiodic words or
$\omega$-words having quasiperiod $q$ can be constructed. In
Section~\ref{sec.codes} it is then proved that the star root of $P_{q}$ is a
suffix code having a bounded delay of decipherability.

This much prerequisites allow us, in Section~\ref{sec.subw} to estimate the
number of subwords of the language $Q_{q}$ of all quasiperiodic words having
quasiperiod $q$. It turns out that $c_{q,1}\cdot \lambda_{q}^{n}\le
f(Q_{q},n)\leq c_{q,2}\cdot \lambda_{q}^{n}$ where $f(Q_{q},n)$ is the number
of subwords of length $n$ of words in $Q_{q}$ and $1\leq \lambda_{q}\leq
t_{P}$ depends on $q$. {}From these results we derive our estimates for the
subword complexity of quasiperiodic infinite words. Finally, we show that, for
every quasiperiod $q$, there is a quasiperiodic \mbox{$\omega$-word} $\xi$
with quasiperiod $q$ whose subword complexity $f(\xi,n)$ meets the upper bound
$c_{q,2}\cdot \lambda_{q}^{n}$.

\section{Notation}
\label{sec.notat}
In this section we introduce the notation used throughout the paper.  By
$\bbbn = \{ 0,1,2,\ldots\}$ we denote the set of natural numbers. Let $X$ be
an alphabet of cardinality $|X| = r\ge 2$. By $X^*$ we denote the set of
finite words on $X$, including the \textit{empty word} $e$, and $X^{\omega}$
is the set of infinite strings ($\omega$-words) over $X$.  Subsets of $X^*$
will be referred to as \textit{languages} and subsets of $X^\omega$ as
\textit{$\omega$-languages}.

For $w\in X^*$ and $\eta\in X^*\cup X^{\omega}$ let $w \cdot{}\eta$ be their
\textit{concatenation}.  This concatenation product extends in an obvious way
to subsets $L \subseteq X^*$ and $B\subseteq X^*\cup X^{\omega}$. For a
language $L$ let $L^* := \bigcup_{i \in \bbbn} L^i$, and by
$L^\omega:=\{w_1\cdots w_i\cdots: w_i\in L\setminus \{e\}\}$ we denote the set
of infinite strings formed by concatenating words in $L$.  Furthermore $|w|$
is the \textit{length} of the word $w\in X^*$ and $\pref B$ is the set of all
finite prefixes of strings in $B\subseteq X^*\cup X^\omega$.  We shall
abbreviate $w\in \pref \eta\ (\eta\in X^*\cup X^\omega)$ by $w\sqsubseteq
\eta$%
.

We denote by $B/w:= \{\eta: w\cdot\eta\in B\}$ the \emph{left derivative} of
the set $B\subseteq X^*\cup X^{\omega}$.  As usual, a language $L\subseteq
X^{*}$ is \textit{regular} provided it is accepted by a finite automaton. An
equivalent condition is that its set of left derivatives $\{L/w: w\in X^{*}\}$
is finite.

The sets of infixes of $B$ or $\eta$ are $\infix B:= \bigcup_{w\in
  X^{*}}\pref{B/w}$ and $\infix \eta:= \bigcup_{w\in X^{*}}\pref{\{\eta\}/w}$,
respectively.  In the sequel we assume the reader to be familiar with basic
facts of language theory.

As usual a language $L\subseteq X^*$ is called a \emph{code} provided
$w_1\cdots w_l= v_1\cdots v_k$ for $w_1,\dots, w_l,$ $v_1,\dots ,v_k\in L$
implies $l=k$ and $w_i=v_i$. 

\section{Quasiperidicity}
\label{sec.quas}
\subsection{General properties}

A finite or infinite word $\eta\in X^{*}\cup X^{\omega}$ is referred to as
\emph{quasiperiodic} with quasiperiod $q\in X^{*}\setminus \{e\}$ provided for
every $j<|\eta|\in \bbbn\cup\{\infty\}$ there is a prefix $u_{j}\sqsubseteq
\eta$ of length $j-|q|<|u_{j}|\leq j$ such that $u_{j}\cdot q\sqsubseteq
\eta$, that is, for every $w\sqsubseteq\eta$ the relation $u_{|w|}\sqsubset
w\sqsubseteq u_{|w|}\cdot q$ is valid.

Let for $q\in X^{*}\setminus \{e\}$, $Q_{q}$ be the set of {quasiperiodic}
words with quasiperiod $q$. Then $\{q\}^{*}\subseteq Q_{q} =Q_{q}^{*}$ and
$Q_{q}\setminus \{e\}\subseteq X^{*}\cdot q\cap q\cdot X^{*}$.

\begin{definition}\upshape{}
  A family $\bigl(w_{i}\bigr)_{i=1}^{\ell}$, $\ell\in \bbbn\cup\{\infty\}$, of
  words $w_{i}\in X^{*}\cdot q$ is referred to as a
  \emph{$q$-chain} 
  provided $w_{1}=q$, $w_{i}\sqsubset w_{i+1}$ and $|w_{i+1}|-|w_{i}|\leq
  |q|$.
\end{definition}
It holds the following.
\begin{lemma}\label{l.Qq}\mbox{ }\\[-18pt]
  \begin{enumerate}
  \item $w\in Q_{q}\setminus\{e\}$ if and only if there is a $q$-chain
    $\bigl(w_{i}\bigr)_{i=1}^{\ell}$ such that $w_{\ell}= w$.
  \item An $\omega$-word $\xi\in X^{\omega}$ is quasiperiodic with quasiperiod
    $q$ if and only if there is a $q$-chain $\bigl(w_{i}\bigr)_{i=1}^{\infty}$
    such that $w_{i}\sqsubset \xi$.
  \end{enumerate}
\end{lemma}
\begin{proof}
  It suffices to show how a family $\bigl(u_{j}\bigr)_{j=0}^{|\eta|-1}$ can be
  converted to a $q$-chain $\bigl(w_{i}\bigr)_{i=1}^{\ell}$ and vice versa.

  Consider $\eta\in X^{*}\cup X^{\omega}$ and let
  $\bigl(u_{j}\bigr)_{j=0}^{|\eta|-1}$ be a family such that $u_{j}\cdot
  q\sqsubseteq \eta$ and $j-|q|<|u_{j}|\leq j$ for $j<|\eta|$.

  Define $w_{1}:= q$ and $w_{i+1}:= u_{|w_{i}|}\cdot q$ as long as
  $|w_{i}|<|\eta|$. Then $w_{i}\sqsubseteq\eta$ and $|w_{i}|< |w_{i+1}|=
  |u_{|w_{i}|}\cdot q|\leq |w_{i}| + |q|$. Thus
  $\bigl(w_{i}\bigr)_{i=1}^{\ell}$ is a $q$-chain with $w_{i}\sqsubseteq
  \eta$.

  Conversely, let $\bigl(w_{i}\bigr)_{i=1}^{\ell}$ be a $q$-chain such that
  $w_{i}\sqsubseteq \eta$ and set
  \[u_{j}:= \max\nolimits_{\sqsubseteq}\bigl\{w': \exists i(w'\cdot q=
  w_{i}\wedge |w'|\leq j)\bigr\}\mbox{ , for }j<|\eta|\,.\] By definition,
  $u_{j}\cdot q\sqsubseteq \eta$ and $|u_{j}|\leq j$. Assume $|u_{j}|\leq
  j-|q|$ and $u_{j}\cdot q=w_{i}$. Then $|w_{i}|\leq j<|\eta|$. Consequently,
  in the $q$-chain there is a successor $w_{i+1}$, $|w_{i+1}|\leq |w_{i}|+
  |q|\leq j+ |q|$. Let $w_{i+1}=w''\cdot q$. Then $u_{j}\sqsubset w''$ and
  $|w''|\leq j$ which contradicts the maximality of $u_{j}$.
\end{proof}
\begin{corollary}\label{c.Qq}
  Let $u\in \pref{Q_q}$. Then there are words $w,w'\in Q_{q}$ such that
  $w\sqsubseteq u\sqsubseteq w'$ and $|u|-|w|, |w'|-|u|\leq |q|$.
\end{corollary}
\begin{corollary}\label{c.Qdelta}
  Let $\xi\in X^{\omega}$. Then the following are equivalent.
  \begin{enumerate}
  \item $\xi$ is quasiperiodic with quasiperiod $q$.
  \item $\pref\xi\cap Q_{q}$ is infinite.
  \item $\pref\xi\subseteq\pref{Q_q}$.
  \end{enumerate}
\end{corollary}
\subsection{A finite generator for quasiperiodic words}
In this part we introduce the finite language $P_q$ which generates the set of
quasiperiodic words as well as the set of quasiperiodic $\omega$-words having
quasiperiod $q$. We investigate basic properties of $P_q$ using simple facts
from combinatorics on words (see e.g.  \cite{Shyr:book}).  We set
\begin{equation}
  \label{eq.Pq}
  P_{q}:= \{v:e\sqsubset v\sqsubseteq q \sqsubset v\cdot q\}\,.
\end{equation}
Then we have the following properties.
\begin{proposition}\label{p.PQ}\mbox{ }\vspace*{-22pt}
  \begin{eqnarray}
    \label{eq.PQ1}
    Q_{q}&=&P_{q}^{*}\cdot q\cup \{e\}\ \subseteq\ P_{q}^{*}\ ,\\
    \label{eq.PQ3} 
    \pref{P_{q}^{*}}&=&\pref{Q_{q}}\ =\ P_{q}^{*}\cdot \pref{q}
  \end{eqnarray}
\end{proposition}
\begin{proof}
  In order to prove Eq.~(\ref{eq.PQ1}) we show that $w_{i}\in P_{q}^{*}\cdot
  q$ for every $q$-chain $\bigl(w_{i}\bigr)_{i=1}^{\ell}$. This is certainly
  true for $w_{1}=q$.  Now proceed by induction on $i$. Let $w_{i}=
  w_{i}'\cdot q\in P_{q}^{*}\cdot q$ and $w_{i+1}= w_{i+1}'\cdot q$. Then
  $w_{i}'\cdot v_{i}= w'_{i+1}$. Now from $w_{i}\sqsubset w_{i+1}$ we obtain
  $e\sqsubset v_{i} \sqsubseteq q\sqsubset v_{i}\cdot q$, that is, $v_{i}\in
  P_{q}$.

  Eq.~(\ref{eq.PQ3}) is an immediate consequence of Eq.~(\ref{eq.PQ1}).
\end{proof}
Corollary~\ref{c.Qdelta} and Proposition~\ref{p.PQ} imply the
following characterisation of $\omega$-words having quasiperiod $q$.
\begin{equation}\label{eq.Pomega}
  \{\xi: \xi\in X^{\omega}\wedge \xi \mbox{ has quasiperiod }q\}= P_{q}^{\omega}
\end{equation}
\begin{proof}
  Since $P_{q}$ is finite, $P_{q}^{\omega}=\{\xi: \xi\in X^{\omega}\wedge
  \pref\xi\subseteq \pref{P_{q}^*}\}$.
\end{proof}
The following property of words in $P_{q}$ is a consequence of the
Lyndon-Sch\"utzen\-ber\-ger Theorem (see \cite{BerstelPerrin1985,Shyr:book}).
\begin{proposition}\label{p.vbar}
  $v\in P_{q}$ if and only if $|v|\leq |q|$ and there is a prefix $\bar
  v\sqsubset v$ such that $q= v^{k}\cdot\bar v$ for $k= \bigl\lfloor
  |q|/|v|\bigr\rfloor$.
\end{proposition}
\begin{proof}
  Sufficiency is clear. Let now $v\in P_{q}$. Then $v\sqsubseteq q\sqsubset
  v\cdot q$. This implies $v^{l}\sqsubseteq q\sqsubset v^{l}\cdot q$ as long
  as $l\leq k$ and, finally, $q\sqsubset v^{k+1}$.
\end{proof}
\begin{corollary}\label{c.vbar}
  $v\in P_{q}$ if and only if $|v|\leq |q|$ and there is a $k'\in\bbbn$ such
  that $q\sqsubseteq v^{k'}$.
\end{corollary}
Now set $q_0:= \min_{\sqsubseteq}P_{q}$.  Then in view of
Proposition~\ref{p.vbar} and Corollary~\ref{c.vbar} we have the following.
\begin{equation}
  q=q_0^{k}\cdot \bar q
  \mbox{ for }k= \bigl\lfloor |q|/|q_0|\bigr\rfloor\mbox{ and some }\bar
  q\sqsubset q_0\,.   
\end{equation}
\begin{corollary}\label{c.q0prim}
  The word $q_0$ is primitive, that is, there are no $u\in X^{*}$ and $n> 1$
  such that $q_0= u^{n}$.
\end{corollary}
\begin{proof}
  Assume $q_0=q_1^l$ for some $l>1$. Then $\bar q= q_1^j\cdot \bar q_1$ where
  $\bar q_1\sqsubset q_1$, and, consequently, $q\sqsubset q_1^{k\cdot l+ j+1}$
  contradicting the fact that $q_0$ is the shortest word in $P_q$.
\end{proof}
\begin{proposition}\label{p.PQcdot}
  \begin{enumerate}
  \item If $v\in P_{q}$ and $w\sqsubseteq q$ then $v\cdot w\sqsubseteq q$ or
    $q\sqsubseteq v\cdot w$.\label{p.PQcdot1}
  \item If $v\in P_{q}$ and $|v|\leq|q|-|q_0|$ then $v=q_0^{m}$ for some
    $m\in\bbbn$.\label{p.PQcdot2}
  \end{enumerate}
\end{proposition}
\begin{proof}
  The first assertion follows from $v\sqsubseteq q\sqsubset v\cdot q$ and
  $v\cdot w\sqsubseteq v\cdot q$.

  For the proof of the second one observe that, by the first item $v\cdot
  q_0\sqsubseteq q$ and $q_0\cdot v\sqsubseteq q$ whence $q_0\cdot v= v\cdot
  q_0$. Thus $q_0$ and $v$ are powers of a common word. Since $q_0$ is
  primitive, the assertion follows.
\end{proof}
\begin{theorem}\label{th.wq0}
  If $v\in P_{q}$ and $w\cdot v\sqsubseteq q$ then $w\in \{q_0\}^{*}$.
\end{theorem}
\begin{proof}If $v\in P_{q}$ then $q_0\sqsubseteq v$. Thus it suffices to
  prove the assertion for $q_{0}$.

  Let $w\cdot q_0\sqsubseteq q= q_0^{k}\cdot \bar q$. Then $w\cdot
  q_0\sqsubseteq q_0^{k+2}$ and, trivially, $q_0\sqsubseteq q_0^{k+2}$. Since
  $|w\cdot q_0|+ |q_0|< |q_0^{k+2}|$, $w\cdot q_0$ and $q_0$ are powers of a
  common word. The assertion follows because $q_0$ is primitive.
\end{proof}
\section{Codes}
\label{sec.codes}
In this section we investigate in more detail the properties of the star root
of $P_{q}$, that is, of the smallest subset $V\subseteq P_{q}$ such that
$V^{*}=P_{q}^{*q}$. It turns out that $\strr{P_q}$ is a suffix code which,
additionally, has a bounded delay of decipherability. This delay is closely
related to the largest power of $q_{0}$ being a prefix of $q$.


According to
\cite{BerstelPerrin1985
} a subset $C\subseteq X^{*}$ is a code of a \emph{delay of decipherability
  $m\in\bbbn$\/} 
if and only if for all $w,w',v_1,\dots, v_m \in C$ and $u\in C^*$ the relation
$w\cdot v_1\cdots v_m \sqsubseteq w'\cdot u$ implies $w = w'$.
Observe that $C\subseteq X^{*}\setminus \{e\}$ is a prefix code, that is,
$w,w',\in C$ and $w\sqsubseteq w'$ imply $w = w'$, if and only if $C$
has 
delay $0$. A subset $C\subseteq X^{*}\setminus \{e\}$ is referred to as a
\emph{suffix code} if no word $w\in C$ is a proper suffix of another word
$v\in C$.

Define now the \emph{star-root} of $P_{q}$:\[\strr{P_q}:= P_{q}\setminus
\bigl(P_{q}^{2}\cdot P_{q}^*\bigr)\] It holds the following.
\begin{equation}\label{eq.starroot}
  \strr{P_q}= \bigl(P_{q}\setminus \{q_0\}^{*}\bigr)\cup\{q_0\}\subseteq 
  \{q_0\}\cup \{v:v\sqsubseteq q\wedge |q_0|+|v|>|q|\}
\end{equation}
\begin{proof}
  First we prove the identity.  The inclusion ``$\subseteq$'' follows from
  $\bigl(P_{q}\setminus \{q_0\}^{*}\bigr)\cup\{q_0\}\subseteq P_{q}\subseteq
  \bigl((P_{q}\setminus \{q_0\}^{*})\cup\{q_0\}\bigr)^{*}$.

  To prove the reverse inclusion assume $\ell>1$ and $v_{1}\cdots v_{\ell}\in
  P_{q}$ for $v_{i}\in P_{q}$. Then $|q_0|\leq |v_{i}|$ and thus
  $|q_0|+|v_{i}|\leq|q|$ for all $i$. According to
  Proposition~\ref{p.PQcdot}.\ref{p.PQcdot2} we have $v_{i}\in \{q_0\}^{*}$
  which shows $P_{q}\cap \bigl(P_{q}^{2}\cdot P_{q}^*\bigr)\subseteq
  \{q_0\}^{*}$.

  The remaining inclusion now follows from
  Proposition~\ref{p.PQcdot}.\ref{p.PQcdot2}.
\end{proof}
Next we are going to show that $\strr{P_q}$ is a suffix code having a bounded
delay of decipherability.
\begin{corollary}\label{c.suffix1}
  $\strr{P_q}$ is a suffix code.
\end{corollary}
\begin{proof}
  Assume $u=w\cdot v$ for some $u,v\in \strr{P_q}\ ,u\ne v$. Then
  Theorem~\ref{th.wq0} proves $w\in \{q_0\}^{*}\subseteq P_{q}$. If $w\ne e$,
  in view of $u\sqsubseteq q$ Proposition~\ref{p.PQcdot}.\ref{p.PQcdot2}
  implies $v\in \{q_0\}^{*}$ and hence $u\in \{q_0\}^{*}$. Thus $u=v=q_{0}$
  contradicting $u\ne v$.
\end{proof}

\begin{theorem}\label{th.delay}
  Let $q=q_0^{k}\cdot\bar q$ where $\bar q\sqsubset q_0$. Then $\strr{P_q}$ is
  a code having a delay of decipherability of at most $k+1$.
\end{theorem}
\begin{proof}
  We have to show that if the words $v\cdot w_{1}\cdots w_{k+1}$ and $v'\cdot
  w'_{1}\cdots w'_{k+1}$, where $v,w_{1},\dots, w_{k+1},$ $v', w'_{1},\dots,
  w'_{k+1}\in \strr{P_q}$ are comparable w.r.t. ``$\sqsubseteq$'' then $v=v'$.

  Without loss of generality, assume $v\sqsubset v'$. Then $|q_0|\leq|v|<
  |v'|\leq|q|$. We have $|w_{i}|,|w_{i}'|\geq |q_0|$. Thus $|w_{1}\cdots
  w_{k+1}|,|w'_{1}\cdots w'_{k+1}|> |q|$. Moreover, according to
  Proposition~\ref{p.PQcdot}.\ref{p.PQcdot1} $q\sqsubseteq w_{1}\cdots
  w_{k+1}$ and $q\sqsubseteq w'_{1}\cdots w'_{k+1}$, whence $v\cdot q\sqsubset
  v'\cdot q$. Then in view of the inequality $|v|+ |q| \geq |v'|+|q_0|$ we
  have $q\sqsupseteq w\cdot q_0$ for the word $w\ne e$ with $v\cdot w= v'$
  and, according to Theorem~\ref{th.wq0} $w\in \{q_0\}^{*}$. This contradicts
  the fact that $\strr{P_q}$ is a suffix code.
\end{proof}
We provide examples that, on the one hand, the bound in Theorem~\ref{th.delay}
cannot be improved and, on the other hand that it is not always
attained. Since for $q=q_{0}^{k},\ k\in\bbbn,$ the code $\strr{P_q}=\{q_{0}\}$
is a prefix code, we consider only non-trivial cases.
\begin{example}\label{ex.delay}\upshape{}
  Let $q:=\mathit{aabaaaaba}$. Then $q_0=\mathit{aabaa}$, $k=1$ and
  $\strr{P_q}=P_{q}=\{\,q_0,\mathit{aabaaaab},q\ \}$ which is a code having a
  delay of decipherability $2$.

  Indeed $
  \begin{array}[t]{rclclr}
    \mathit{aabaaaabaa}&=&q_0\cdot q_0&\sqsubseteq& q\cdot q_0 \mbox{\quad
      or}&\hspace{37mm}\\ 
    \mathit{aabaaaabaa}&=&q_0\cdot q_0&\sqsubseteq& \mathit{aabaaaab}\cdot
    q_0\,.&
  \end{array}
  $\\[-11pt]
  \mbox{ }\hfill\qed
\end{example}

\medskip{}

Moreover $q\cdot q_0\notin Q_q$. Thus our Example~\ref{ex.delay} shows also
that $q\cdot P_q^*$ need not be contained in $Q_q$.
\begin{example}\label{ex.delay1}\upshape{}
  Let $q:=\mathit{aba}$. Then $k=1$ and $P_{q}=\{\mathit{ab},\mathit{aba}\}$
  is a code having a delay of decipherability~$1$.\hspace*{1pt}\hfill\qed
\end{example}

\section{Subword Complexity}\label{sec.subw}
In this section we investigate the subword complexity of the language
$Q_q$. To this end we derive general relations between the numbers of words of
a certain length for regular languages, their prefix- and their
infix-languages.  Then using elementary methods of the theory of formal power
series (cf. \cite{BerstelPerrin1985,book:SalomaaSoi}) we estimate values
characterising the exponential growth of the family $(|\infix{Q_q}\cap
X^n|)_{n\in\bbbn}$.

We start with some prerequisites on the number of subwords of regular
star-languages.
\begin{lemma}\label{l.L}
  If $L\subseteq X^*$ is a regular language then there is a $k\in\bbbn$ such
  that
  \begin{equation}\begin{array}{rcccl}
      |L\cap X^{n}|&\le&|\pref{L}\cap X^n|&\le&\sum_{i=0}^{k}|L\cap X^{n+i}|\\[4pt]
      |\pref{L}\cap X^n|&\le&|\infix{L}\cap X^n|&\le& k\cdot |\pref{L}\cap X^n|
    \end{array}
  \end{equation} 
\end{lemma}
As a suitable $k$ one may choose the number of states of an automaton
accepting the language $L\subseteq X^*$.

Moreover, Corollary~4 of \cite{St85:probl} shows that for every regular
language $L\subseteq X^{*}$ there are constants $c_{1},c_{2}>0$ and a
$\lambda\ge 1$ such that
\begin{equation}\label{eq.St85}
  c_{1}\cdot \lambda^n\le |\pref{L^*}\cap X^n|\le c_{2}\cdot \lambda^n\,.
\end{equation}
A consequence of Lemma~\ref{l.L} is that Eq.~(\ref{eq.St85}) holds also (with
constant $k\cdot c_{2}$ instead of $c_{2}$) for $\infix{L^*}$.
\subsection{The subword complexity of $Q_q$}

It is now our task to estimate the value $\lambda_{q}$ which satisfies
$c_{1}\cdot \lambda_{q}^n\le |\infix{P_{q}^*}\cap X^n|\le k\cdot c_{2}\cdot
\lambda_{q}^n$. Following Lemma~\ref{l.L} and Eqs.~(\ref{eq.St85}) and
(\ref{eq.PQ3}) it holds
\begin{equation}\label{eq.radius}
  \lambda_{q} = \limsup_{n\to\infty} \sqrt[n\,]{|P_{q}^*\cap X^{n}|}
\end{equation}
which is the inverse of the convergence radius
$\mathsf{rad}\,\mathfrak{s}^{*}_{q}$ of the power series
$\mathfrak{s}^{*}_{q}(t):=\sum_{n\in\bbbn}|P_{q}^*\cap X^{n}|\cdot t^{n}$ (the
structure generating function of the language $P_{q}^*$).

If $|q_{0}|$ divides $|q|$ then $P_{q}^*= \{q_0\}^{*}$ whence
$\lambda_{q}=1$. Therefore, in the following considerations we may assume that
$|q|/|q_{0}|\notin \bbbn$.

Since $\strr{P_{q}}$ is a code, we have $\mathfrak{s}^{*}_{q}(t)= \frac{1}{1-
  \mathfrak{s}_{q}(t)}$ where $\mathfrak{s}_{q}(t):=\sum_{v\in
  \sqrt[*\,]{P_{q}}}t^{|v|}$ is the structure generating function of the
finite language $\strr{P_{q}}$. Thus the convergence radius
$\mathsf{rad}\,\mathfrak{s}^{*}_{q}$ is the smallest root of
$1-\mathfrak{s}_{q}(t)$. It is readily seen that this root is positive. So
$\lambda_{q}$ is the largest positive root of the reversed
polynomial\footnote{If $|q_{0}|$ divides $|q|$ we have $\mathfrak{p}_{q}(t)=
  t^{|q_{0}|}-1$ instead.}  $\mathfrak{p}_{q}(t):= t^{|q|}-\sum_{v\in
  \sqrt[*\,]{P_{q}}}t^{|q|-|v|}$. Summarising these observations we obtain
the following.
\begin{lemma}\label{l.growth}
  Let $q\in X^{*}\setminus \{e\}$. Then there are constants
  $c_{q,1},c_{q,2}>0$ such that the structure function of the language
  $\infix{Q_{q}}$ satisfies \[c_{q,1}\cdot \lambda_q^{n}\leq|
  \infix{Q_{q}}\cap X^{n}|\leq c_{q,2}\cdot \lambda_q^{n}\] where $\lambda_q$
  is the largest (positive) root of the polynomial $\mathfrak{p}_{q}(t)$.
\end{lemma}
\textit{Remark. }One could prove Lemma~\ref{l.growth} by showing that, for
each polynomial $\mathfrak{p}_{q}(t)$, its largest (positive) root has
multiplicity $1$.  Referring to Corollary~4 of \cite{St85:probl} (see
Eq.~(\ref{eq.St85})) we avoided these more detailed considerations of a
particular class of polynomials.

In order to facilitate the search for the maximum of the values $\lambda_{q}$
we may restrict our considerations to the case when $|q_{0}| > |q|/2$.
\begin{lemma}
  If $|q_{0}|$ does not divide $|q|$ and the language $P_{q}^{*}$ is maximal
  w.r.t. ``$\subseteq$'' in the class $\bigl\{P_{q'}^{*}:q'\in X^{*}\setminus
  \{e\}\bigr\}$ then $|q_{0}| > |q|/2$.
\end{lemma}
\begin{proof}
  If $|q|/|q_0|\notin \bbbn$ and $|q_{0}|\le |q|/2$ we have $q=q_{0}^{k}\cdot
  \bar q$ for $k\ge2$ and $e\ne \bar q\sqsubset q_{0}$. Then, obviously
  $P_{q}^{*}\subset P_{q'}^{*}$ for $q':= q_{0}\cdot \bar q$.
\end{proof}
{}From $|q_{0}| > |q|/2$ we obtain that $\mathfrak{p}_q(t)$ has the form
$t^{|q|}- \sum_{i\in M}t^{i}$ where $0\in M\subseteq
\{j:j<\frac{|q|}{2}\}$. In \cite{Poll09} the following properties were
derived.
\begin{lemma}\label{l.poly}
  Let $\mathcal{P}:=\bigl\{t^n- \sum_{i\in M} t^i: n\ge1\wedge 0\in M\subseteq
  \{j:j\leq\frac{n-1}{2}\}\bigr\}$. Then
  \begin{enumerate}
  \item for every $n\geq 1$ the polynomial $t^n-\sum_{i=0}^{\lfloor
      \frac{n-1}{2}\rfloor}t^i$ has the largest positive root among all
    polynomials of degree $n$ in $\mathcal{P}$, and
  \item the polynomials $t^{3}-t-1$ and
    $t^{5}-t^{2}-t-1=(t^{2}+1)\cdot(t^{3}-t-1)$ have the largest positive
    roots among all polynomials in $\mathcal{P}$.
  \end{enumerate}
\end{lemma}

Two remarks are in order here. 
\begin{enumerate}
\item It holds
$\mathfrak{p}_{a^{n}ba^{n}}(t)=t^{2n+1}- \sum_{i=0}^nt^i$ and
$\mathfrak{p}_{a^{n}b^2a^{n}}(t)=t^{2n+2}- \sum_{i=0}^nt^i$, so for all
degrees $\ge 1$ there are polynomials of the form $\mathfrak{p}_q(t)$ in
$\mathcal{P}$.
\item The positive root $t_{P}$ of $\mathfrak{p}_{aba}(t)=t^{3}-t-1$ (or of
$\mathfrak{p}_{a^2ba^2}(t)$) is known as the smallest Pisot-Vijayaraghavan
number, that is, a positive root $>1$ of a polynomial with integer
coefficients all of whose conjugates have modulus smaller than $1$.
\end{enumerate}
Before proceeding to the proof of Lemma~\ref{l.poly} we recall that
the polynomials $p(t)\in\mathcal{P}$ have the following easily verified
property.
\begin{equation}\label{eq.poly2}
  \mbox{If }\varepsilon>0\mbox{ and }p(t') \ge 0\mbox{ for some }t'>0\mbox{
    then }p((1+\varepsilon)\cdot t')>0\,. 
\end{equation}
Since $p(0)=-1<0$ for $p(t)\in\mathcal{P}$, Eq.~(\ref{eq.poly2}) shows that
once $p(t')\ge 0,\ t'>0$ the polynomial $p(t)$ has no further root in the
interval $(t',\infty)$.

\bigskip

\begin{proof}
  Using Eq.~(\ref{eq.poly2}) the first assertion is easy to verify.

  To show the second one it suffices to show that $p_n(t_{P})>0$ for every
  polynomial of the form $p_n(t):= t^n - \sum_{i=0}^{\lfloor \frac{n-1}{2}
    \rfloor} t^i$ other than $t^{3}-t-1$ or $t^{5}-t^{2}-t-1$.

  For degrees $n=1,2$ or $n=4$ this is readily seen.

  Now we proceed by induction on $n$. To this end we observe the following
  properties of the family $(p_n(t))_{n\ge 1}$.
  \begin{equation}
    p_{n+2}(t)-p_{n}(t)=t^{n+2}-t^{n}- t^{\lfloor \frac{n+1}{2} \rfloor}\mbox{
      for }n\ge 3 
  \end{equation}
  {}From this one easily obtains that $p_{n+2}(t_P)-p_{n}(t_P)= t_P^{n-1}-
  t_P^{\lfloor \frac{n+1}{2} \rfloor}>0$ for $n\ge 4$, and the assertion
  follows by induction.
\end{proof}

\subsection{The subword complexity of $\omega$-words}
\label{sec.subomega}
Having derived the results on the the subword complexity of quasiperiodic
words we are now in a position to contribute to an answer to Question~2 in
\cite{DBLP:journals/eatcs/Marcus04} by deriving tight upper bounds on the
subword complexity of quasiperiodic infinite words.

To this end we recall that $\infix\xi\subseteq \infix{Q_{q}}$ for every
$\omega$-word $\xi$ with quasiperiod $q$. Thus we obtain the following upper
bound.
\begin{lemma}
  If $\xi\in X^\omega$ is quasiperiodic with quasiperiod $q$ then
  $f(\xi,n)=|\infix\xi\cap X^{n}|\leq c\cdot \lambda_{q}^{n}$ for a suitable
  constant $c>0$ not depending on $\xi$.
\end{lemma}
Following the proof of Proposition 5.5 in \cite{DBLP:journals/iandc/Staiger93}
it can be shown that this upper bound is tight.
\begin{lemma}
  For every quasiperiod $q\in X^*\setminus\{e\}$ there is a $\xi\in
  P_q^\omega$ such that $c_{q,1}\cdot \lambda_{q}^{n}\leq
  f(\xi,n)=|\infix\xi\cap X^{n}|$.
\end{lemma}
Here $c_{q,1}$ is the constant mentioned in Lemma~\ref{l.growth}.
\begin{proof}
  Let $P_q^*=\{v_0,v_{1},v_2\ldots \}$ and define $\xi:=
  \prod_{i\in\bbbn}v_{i}$. Then obviously $\infix\xi=\infix{P_q^*}=\infix{Q_q}$.
\end{proof}

An over-all upper bound on the subword complexity of quasiperiodic
$\omega$-words now follows from Lemma~\ref{l.poly}.
\begin{theorem}\label{th.upper}
  There is a constant $c>0$ such that for every quasiperiodic $\omega$-word
  $\xi\in X^\omega$ there is an $n_\xi\in \bbbn$ such that
  $f(\xi,n)=|\infix\xi\cap X^{n}|\leq c\cdot t_{P}^n$ for all $n\ge n_\xi$.
\end{theorem}
We conclude this section with the following remark.

\noindent\textit{Remark. }Theorem~\ref{th.upper}
is independent of the size of the alphabet $X$. And indeed, quasiperiodic
$\omega$-words of maximal subword complexity have quasiperiods of the form
$aba$ or $aabaa$, $a,b\in X,\ a\ne b$ (see the remark after
Lemma~\ref{l.poly}), thus consist of only two different letters.

\section{Concluding Remark}
In the present paper we investigated the maximally achievable subword
complexity for quasiperiodic infinite words. It should be mentioned that using
results of \cite{DBLP:journals/iandc/Staiger93} the bounds obtained here can
be extended to the Kolmogorov complexity of infinite words.

In \cite[Section~5]{DBLP:journals/iandc/Staiger93} the asymptotic subword
complexity of an $\omega$-word $\xi\in X^\omega$ was introduced as
$\tau(\xi):= \lim_{n\to\infty}\frac{\log_{|X|}|\infix\xi\cap X^{n}|}{n}$
and it was shown that $\tau$ is an upper bound to the asymptotic upper and
lower Kolmogorov complexities of infinite words:
\[\underline\kappa(\xi)\le \kappa(\xi)\le\tau(\xi)\,.\]
Moreover, from the results of \cite[Section~4]{DBLP:journals/iandc/Staiger93}
it follows that for every quasiperiodic word $q$ there is a $\xi\in
P_q^\omega$ such that $\underline\kappa(\xi)=\tau(\xi)=\log_{|X|}\lambda_q$,
that is, a quasiperiodic $\omega$-word having quasiperiod $q$ of maximally
possible asymptotic (lower) Kolmogorov complexity.  
\bibliographystyle{eptcs} 

\begin{thebibliography}{Mar04}
\bibitem[AS03]{book:AlloucheSha}
Jean-Paul Allouche and Jeffrey Shallit.
\newblock {\em Automatic sequences}.
\newblock Cambridge University Press, Cambridge, 2003.
\newblock Theory, applications, generalizations.

\bibitem[BK03]{DBLP:journals/eatcs/BerstelK03}
Jean Berstel and Juhani Karhum{\"a}ki.
\newblock Combinatorics on words: a tutorial.
\newblock {\em Bulletin of the EATCS}, 79:178--228, 2003.

\bibitem[BP85]{BerstelPerrin1985}
Jean Berstel and Dominique Perrin.
\newblock {\em Theory of codes}, volume 117 of {\em Pure and Applied
  Mathematics}.
\newblock Academic Press Inc., Orlando, FL, 1985.

\bibitem[LR04]{DBLP:journals/eatcs/LeveR04}
Florence Lev{\'e} and Gw{\'e}na{\"e}l Richomme.
\newblock Quasiperiodic infinite words: Some answers (column: Formal language
  theory).
\newblock {\em Bulletin of the EATCS}, 84:128--138, 2004.

\bibitem[LR07]{DBLP:journals/tcs/LeveR07}
Florence Lev{\'e} and Gw{\'e}na{\"e}l Richomme.
\newblock Quasiperiodic {S}turmian words and morphisms.
\newblock {\em Theor. Comput. Sci.}, 372(1):15--25, 2007.

\bibitem[Mar04]{DBLP:journals/eatcs/Marcus04}
Solomon Marcus.
\newblock Quasiperiodic infinite words (columns: Formal language theory).
\newblock {\em Bulletin of the EATCS}, 82:170--174, 2004.

\bibitem[Pol09]{Poll09}
Ronny Polley.
\newblock Subword complexity of infinite words.
\newblock Diploma thesis, Martin-Luther-Universit\"at Halle-Wittenberg,
  Institut f\"ur Informatik, Halle, 2009.

\bibitem[Shy01]{Shyr:book}
Huei-Jan Shyr.
\newblock {\em Free Monoids and Languages}.
\newblock Hon Min Book Company, Taichung, third edition, 2001.

\bibitem[SS78]{book:SalomaaSoi}
Arto Salomaa and Matti Soittola.
\newblock {\em Automata-theoretic aspects of formal power series}.
\newblock Springer-Verlag, New York, 1978.
\newblock Texts and Monographs in Computer Science.

\bibitem[Sta85]{St85:probl}
Ludwig Staiger.
\newblock The entropy of finite-state {$\omega$}-languages.
\newblock {\em Problems Control Inform. Theory/Problemy Upravlen. Teor.
  Inform.}, 14(5):383--392, 1985.

\bibitem[Sta93]{DBLP:journals/iandc/Staiger93}
Ludwig Staiger.
\newblock {K}olgomorov complexity and {H}ausdorff dimension.
\newblock {\em Inf. Comput.}, 103(2):159--194, 1993.

\bibitem[Sta97]{DBLP:conf/csl/Staiger97}
Ludwig Staiger.
\newblock Rich $\omega$-words and monadic second-order arithmetic.
\newblock In Mogens Nielsen and Wolfgang Thomas, editors, {\em CSL}, volume
  1414 of {\em Lecture Notes in Computer Science}, pages 478--490. Springer,
  1997.
\end{thebibliography}

\end{document}